%% file: main.tex
\documentclass{article}

\usepackage{arxiv}

\usepackage[utf8]{inputenc} 
\usepackage[T1]{fontenc}    
\usepackage{hyperref}       
\usepackage{url}            
\usepackage{booktabs}       
\usepackage{amsfonts}       
\usepackage{nicefrac}       
\usepackage{microtype}      
\usepackage{lipsum}
\usepackage{graphicx}

\usepackage{algorithm}
\usepackage{algorithmic}
\usepackage{appendix}
\usepackage[most]{tcolorbox}
\usepackage{xcolor}
\usepackage{booktabs}
\usepackage{amsmath,amsthm,amssymb}
\usepackage{bbold}
\newtheorem{definition}{Definition}
\newtheorem{theorem}{Theorem}

\usepackage{subfigure}
\usepackage{bbding}
\usepackage{multirow}

\usepackage{newfloat}
\usepackage{listings}
\usepackage{threeparttable}

\definecolor{promptbackground}{HTML}{F8F8F8} %
\definecolor{promptborder}{HTML}{DDDDDD}
\definecolor{prompttitle}{HTML}{444444}

\title{SoftPipe: A Soft-Guided Reinforcement Learning Framework for Automated Data Preparation}

\newcommand{\szucsse}{$^{1}$}
\newcommand{\szuai}{$^{2}$}
\newcommand{\osaka}{$^{3}$}
\author{
 Jing Chang \szucsse \\
  \texttt{2350273007@email.szu.edu.cn} \\
   \And
 Chang Liu \szucsse \\
  \texttt{2300271084@email.szu.edu.cn} \\
  \And
 Jinbin Huang \szuai \\
  \texttt{jbhuang@comp.hkbu.edu.hk} \\
  \And
 Shuyuan Zheng \osaka \\
  \texttt{zheng@ist.osaka-u.ac.jp} \\
  \And
 Rui Mao \szucsse \\
  \texttt{mao@szu.edu.cn} \\
  \And
 Jianbin Qin \szucsse \textsuperscript{\Envelope} \\
  \texttt{qinjianbin@szu.edu.cn} \\
  \And \\
    \szucsse College of Computer Science and Software Engineering, Shenzhen University \\ \szuai School of Artificial Intelligence, Shenzhen University \\ \osaka Graduate School of Information Science and Technology, Osaka University
}

\begin{document}
\maketitle

\begin{abstract}
Data preparation is a foundational yet notoriously challenging component of the machine learning lifecycle, characterized by a vast combinatorial search space. While reinforcement learning (RL) offers a promising direction, state-of-the-art methods suffer from a critical limitation: to manage the search space, they rely on rigid ``hard constraints'' that prematurely prune the search space and often preclude optimal solutions. To address this, we introduce SoftPipe, a novel RL framework that replaces these constraints with a flexible ``soft guidance'' paradigm. SoftPipe formulates action selection as a Bayesian inference problem. A high-level strategic prior, generated by a Large Language Model (LLM), probabilistically guides exploration. This prior is combined with empirical estimators from two sources through a collaborative process: a fine-grained quality score from a supervised Learning-to-Rank (LTR) model and a long-term value estimate from the agent's Q-function. Through extensive experiments on 18 diverse datasets, we demonstrate that SoftPipe achieves up to a 13.9\% improvement in pipeline quality and 2.8$\times$ faster convergence compared to existing methods.
\end{abstract}


\section{Introduction}\label{sec:intro}
 \input{introduction_new}

\section{Related Works}\label{sec:related}
\input{related}

\section{Problem Definition}\label{sec:prelimnaries}
\input{preliminaries}

\section{The SoftPipe Framework Overview}\label{sec:overview}
\input{overview}

\section{Methodology}\label{sec:proposed}
\input{methods}

\section{Experiments}\label{sec:exp}
\input{experiments}

\section{Conclusion}\label{sec:conclusion}
In this paper, we introduced SoftPipe, a novel reinforcement learning framework for automated data preparation that replaces the rigid "hard constraints" of existing methods with flexible "soft guidance." By formulating action selection as a Bayesian inference problem, SoftPipe synergistically combines strategic knowledge from LLMs, immediate quality assessments from LTR models, and long-term value estimates from RL to navigate the vast combinatorial search space more effectively. The wall-robot-nav case study exemplifies this advantage: SoftPipe discovered a superior pipeline combining multiple operators of the same type, a configuration inaccessible to hard-constraint methods like CtxPipe.
SoftPipe's success suggests that the principle of soft guidance over hard constraints has broader implications for combinatorial optimization in AutoML. Future work could explore enhanced LLM prompting strategies, dynamic pipeline lengths, and applications to other AutoML challenges.

\bibliographystyle{unsrt}  
\bibliography{main}  

\appendix
\input{Appendix}

\end{document}

%% file: introduction_new.tex
\label{sec:introduction}
Data preparation is a foundational, yet notoriously labor-intensive, component of the machine learning (ML) lifecycle~\cite{pecan2024data}. The process of transforming raw, often messy data into a clean, well-structured format suitable for model training is not a single action but a multi-stage workflow encompassing tasks like data cleaning, integration, transformation, feature engineering, and feature selection. The quality of the final ML model is profoundly dependent on the quality of this preparatory pipeline~\cite{challange2025,John2025data}. The task's primary difficulty lies in discovering an optimal sequence of operations from a vast library of potential operators, subject to a severe combinatorial explosion. To illustrate, with merely 24 available operators and pipelines of a modest length of 8, the search space contains over $O(24^8) \approx 1.1 \times 10^{11}$ possible configurations. 

Reinforcement Learning (RL) has emerged as a particularly promising paradigm for this challenge, as its formulation as a sequential decision-making process naturally aligns with pipeline construction~\cite{berti2019learn2clean,heffetz2020deepline,shang2019democratizing,yang2021auto,koka2025cleansurvival,chen2025advancing}.  However, a close examination of existing state-of-the-art methods reveals a critical, unaddressed limitation: to manage the vast search space, they implicitly rely on ``hard constraints'' that rigidly structure the pipeline construction process. These constraints, while simplifying the search, prematurely prune the solution space and often preclude optimal pipelines.


A prime example is CtxPipe~\cite{ctxpipe}, a leading RL-based framework. While claiming flexibility, its implementation imposes a crucial hard constraint: It prevents the selection of multiple operators belonging to the same semantic type. These types, detailed in Table~\ref{tab:components}, span the entire data preparation workflow. This seemingly minor restriction has major consequences. On the wall-robot-nav dataset as Table~\ref{tab:avila-compare} listed, CtxPipe's constrained search achieves an accuracy of 0.946. Similarly, HAI-AI~\cite{chen2023haipipe} and TPOT~\cite{olson2016evaluation}, also fall short, achieving 0.896 and 0.941 respectively. Compared to an approximated optimal solution (ES*) discovers a superior pipeline\textit{-}[QuantileTransformer, StandardScaler, PolynomialFeatures]-reaching 0.962 accuracy. This pipeline's effectiveness hinges on combining two `Feature Preprocessing' operators, a structure inaccessible to CtxPipe's search strategy, thus demonstrating the weakness of such hard constraints~\footnote{All data presented are directly drawn from original paper.}.


\begin{table}[h!]
\centering
\caption{Selected data preparation operators and their types}
\label{tab:components}
\begin{tabular}{|c|l|}
\hline
\textbf{Type} & \textbf{Operator (ID)}                   \\ \hline
Imputer       & Mean(1), Median(2), Most Frequent(3) \\ \hline
Encoder       & Label(4), One-hot(5)                 \\ \hline
\begin{tabular}[c]{@{}c@{}}Feature \\ Preprocessing\end{tabular} &
  \begin{tabular}[c]{@{}l@{}}Min Max Scaler(6), Max Abs Scaler(7),\\ Robust Scaler(8), Standard Scale(9),\\ Quantile Transformer(10), Power \\ Transformer(12), Log Transformer(11), \\ Normalizer(13), K-bins Discretizer(14)\end{tabular} \\ \hline
\begin{tabular}[c]{@{}c@{}}Feature \\ Engineering\end{tabular} &
  \begin{tabular}[c]{@{}l@{}}Polynomial Features(15), Interaction \\Features(16), PCA AUTO(17), \\ PCA LAPACK(18), PCA ARPACK(19),\\ Incremental PCA(20), Kernel PCA(21), \\ Truncated SVD(22),\\ Random Trees Embedding(23)\end{tabular} \\ \hline
Selection     & Variance Threshold(24)                 \\ \hline
\end{tabular}
\end{table}

\begin{table}[htbp]
    \centering
    \caption{\textit{wall-robot-nav} dataset performance Comparison} 
    \begin{tabular}{lcl}
    \toprule
        Approach & Accuracy & Pipelines\tnote{1} \\ \midrule
        CtxPipe & 0.946 & [6, 23] \\
        HAI-AI & 0.896 & [6, 15, 24] \\
        TPOT & 0.941 & [14, 15] \\
        ES\textsuperscript{*} & 0.962 & [10, 9, 15] \\
        \bottomrule
    \end{tabular}
    \begin{tablenotes}
        \footnotesize
        \item[1] The pipelines are represented as sequences of operator id, and (id, operator) mappings are listed in Table~\ref{tab:components}.
    \end{tablenotes}
    \label{tab:avila-compare}
\end{table}

In summary, the existing methods discussed above with hard constraints reveals several limitations that block the search of optimal data preparation pipelines: (1)\textit{ Rigid search space pruning:} Hard constraints restrictions exclude potentially optimal solutions from consideration; (2) \textit{Premature convergence:} By enforcing strict rules early in the search process, these methods may converge to suboptimal local minima without exploring promising alternative pathways.(3) \textit{Limited adaptability:} Hard constriants lack the flexibility to adapt based on changing data characteristics during the search process.

To address these limitations, we introduce SoftPipe, a framework that replaces hard constraints with ``soft guidance''. Our approach formulates action selection as a Bayesian inference problem. Instead of rigid rules, we use a Large Language Model (LLM) to generate a high-level strategic prior, which probabilistically guides the exploration. This prior is then synergistically combined with empirical evidence from a fine-grained Learning-to-Rank (LTR) model and the agent's own long-term value estimates (Q-function). This principled fusion allows SoftPipe to benefit from strategic guidance while maintaining the flexibility to override it when evidence warrants, enabling the discovery of unconventional yet optimal pipelines. In summary, our contributions are: 

\begin{itemize}
    \item We identify and expose a critical limitation in current AutoDP methods: their reliance on rigid ``hard constraints'' that prematurely prune the search space and prevent the discovery of optimal pipelines.
    \item We propose SoftPipe, a novel RL framework that enables flexible exploration via a ``soft guidance'' mechanism grounded in Bayesian inference that facilitates the collaboration of a strategic LLM prior, a fine-grained LTR score, and an empirical RL value function.
    \item We demonstrate significant improvements, achieving up to 13.9\% higher pipeline quality and 2.8× faster convergence compared to state-of-the-art methods across 18 diverse datasets.
\end{itemize}

%% file: related.tex

\paragraph{Automated Data Preparation (AutoDP)} AutoDP has emerged as a critical component in AutoML systems aiming to streamline preprocessing workflows. Early approaches like Auto-Sklearn~\cite{10.5555/2969442.2969547} integrated rule-based preprocessings along side hyperparameter optimization, while TPOT~\cite{olson2016evaluation} extended the process by combining genetic algorithms to evolve feature engineering and transformation steps. Interactive systems like Data Civilizer~\cite{rezig2019datacivilizer} and DynaML~\cite{DynaML} provide visual platforms for data curation. Commercial systems like H2O.ai~\cite{h2o} further automate data preparation by incorporating automated feature engineering templates. These cases demostrate the awareness that data preparation strongly impacts tasks in data science.

\paragraph{Reinforcement Learning for AutoDP} Recent studies exploreReinforcement learning (RL) and neural networks (NN) for AutoDP. RL-based models, such as Learn2Clean~\cite{berti2019learn2clean}, model data as a sequential decision process~\cite{zhang2025causalcomrl,yu2018towards,khurana2018feature}. Meanwhile, MetaPrep~\cite{zagatti2021metaprep} combines meta-learning with RL to adaptively select preprocessing steps for new tasks, reducing search space via learned priors. In the context of AutoDP, RL has been applied to specific sub-problems like feature selection~\cite{chen2023haipipe,ctxpipe,heffetz2020deepline}. However, the RL agent may face scalability challenges and lack of inherent structure guidance, causing possible overfitting risks.

\paragraph{LLMs for Code and Planning} Recently, Large Language Models (LLMs) have shown remarkable capabilities in generating code and plans for data science tasks~\cite{chen2021evaluating,chen2024can,gao2023automl,ghosh2024jupyter}. These approaches typically operate in an ``open-loop'' fashion: the LLM receives a prompt and generates a static pipeline or script. While impressive, this paradigm lacks the ability to adapt based on feedback from executing the operations.


%% file: preliminaries.tex
\label{sec:problem_def}

The goal of Automated Data Preparation (AutoDP) is to find an optimal sequence of data transformation operators, called a pipeline, that maximizes the performance of a downstream machine learning model. We formalize this task as follows:

\begin{definition}[Data Preparation Pipeline]
A data preparation pipeline $P$ is a sequence of operators $P = (o_1, o_2, \dots, o_T)$ of a fixed length $T$, where each operator $o_t \in \mathcal{O}$ is drawn from a global library of available operators. Applying a pipeline $P$ to a raw dataset $D_{raw}$ produces a transformed dataset $D_{P}$.
\end{definition}

The core challenge lies in the combinatorial nature of the search space. With $|\mathcal{O}|$ available operators and pipelines of length $T$, the search space contains $|\mathcal{O}|^T$ possible configurations. To navigate this space, we model the pipeline construction process as a sequential decision-making problem.


\begin{definition}[Pipeline Construction as an MDP]
We model the task as a Markov Decision Process (MDP), defined by the tuple $\mathbf{M} = (\mathcal{S}, \mathcal{A}, \mathcal{P}, \mathcal{R}, \gamma)$.
\begin{itemize}
    \item $\mathcal{S}$ is the state space, where each state $s \in \mathcal{S}$ represents the dataset's characteristics (e.g., meta-features) after applying a partial pipeline.
    \item $\mathcal{A}$ is the action space, consisting of all available data preparation operators $\mathcal{O}$. At each step, selecting an action $a \in \mathcal{A}$ means applying operator $a$ to the current dataset.
    \item $\mathcal{P}(s'|s, a)$ is the state transition probability function, applying operator $a$ to state $s$ produces state $s'$.
    \item $\mathcal{R}(s, a)$ is the reward function, which we define as the performance of a model trained on the final transformed dataset.
    \item $\gamma \in [0, 1]$ is the discount factor. 
\end{itemize}
The objective is to learn a policy $\pi(a|s)$ that generates a pipeline, maximizing the expected cumulative reward.
\end{definition}

While the MDP formulation is standard, existing methods often introduce hard constraints to simplify the search, thereby sacrificing flexibility. To formally describe this limitation, we define operator types and the resulting constrained search.

\begin{definition}[Operator Types and Constrained Search]
The global operator library $\mathcal{O}$ can be partitioned into a set of $K$ disjoint semantic types $\mathcal{M} = \{m_1, \dots, m_K\}$ (e.g., Table~\ref{tab:components}). where $\mathcal{O} = \bigcup_{k=1}^{K} \mathcal{O}{m_k}$ and $\mathcal{O}{m_i} \cap \mathcal{O}_{m_j} = \emptyset$ for $i \neq j$. 

A constrained search strategy restricts the policy at state $s$ to a subset of actions $\mathcal{O}_c(s) \subset \mathcal{O}$ based on pre-defined rules. For instance, a common hard constraint forbids selecting an operator of type $m$ if an operator of the same type already exists in the partially built pipeline:
$\mathcal{O}_c(s) = \{o \in \mathcal{O} : \text{type}(o) \notin \text{types}(\mathcal{P}_{\text{partial}})\}$
\end{definition}

\noindent\textbf{Problem Statement.}
\label{prob:main}
Given a raw dataset $D_{raw}$ and a downstream task $\tau$, our goal is to find an optimal pipeline $P^*$ that maximizes performance when applied to $D_{raw}$:
\begin{equation}
    P^* = \arg\max_{P \in \Pi} \text{Performance}(\tau, P(D_{raw}))
\end{equation}
\noindent where $P(D_{raw})$ denotes the dataset resulting from applying pipeline $P$ to $D_{raw}$.
The key challenge is that a constrained search space $\Pi_c \subset \Pi$, employed by existing methods, may not contain $P^*$. Therefore, the problem is to devise a search strategy that can efficiently explore the full space $\Pi$ by replacing hard constraints with flexible, soft guidance.

%% file: overview.tex
To address the challenge of flexible pipeline search, we introduce SoftPipe, a framework that replaces hard constraints with a flexible, soft guidance mechanism. The core idea is to formulate action selection as a Bayesian inference problem, where we principledly fuse high-level strategic knowledge with task-specific empirical value estimates. 
The architecture of SoftPipe is visualized in Figure~\ref{fig:architecture}. The framework operates through three conceptual layers:

\begin{enumerate}
    \item \textit{Strategic Planner:} At the top layer, a Large Language Model (LLM) assesses the current dataset state ($s_t$) and produces a probabilistic distribution over operator types. This serves as the strategic Bayesian \textit{prior}, guiding the search towards promising areas.

    \item \textit{Model Estimators:} The middle layer provides the estimators that forms the Bayesian \textit{likelihood}. It consists of two complementary components: a Reinforcement Learning (RL) Q-model to estimate the long-term, sequential value of an action, and a pre-trained Learning-to-Rank (LTR) model to provide an immediate score for its quality.

    \item \textit{Collaborative Policy:} The bottom layer fuses the strategic prior from the LLM with the empirical estimators from the RL and LTR models. This is achieved via Bayesian inference to produce a final, flexible policy distribution. The action for the next step of the pipeline is then sampled from this posterior distribution.
\end{enumerate}
The three layers collaborate through a continuous feedback loop: Raw data first enters the system where meta-features are extracted to represent the current state $s_t$, then feeds into all three components: $P_{\text{LLM}}(m|s_t)$, $Q(s_t, a_t)$ and $r_{\text{LTR}}(a_t|s_t)$, producing a unified policy distribution in the collaboration layer. After sampling an action at from this distribution and applying the corresponding operator, the system trnsitions to a new state $s_{t+1}$, and the process repeats iteratively until a complete pipeline is constructed.

This colloaboation design ensures each decision benefits from domain knowledge and learned experience, enabling discovery of optimal pipelines. Having established this overview, we now present the detailed theoretical foundation and implementation of each component.

\begin{figure}[t]
    \centering
    \includegraphics[width=0.6\linewidth]{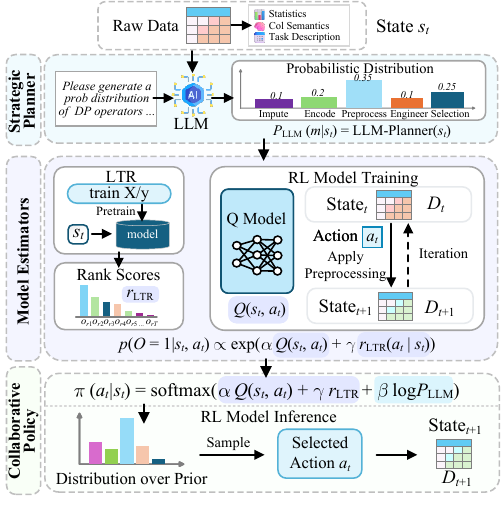}
    \caption{The SoftPipe Framework.} 
    \label{fig:architecture}
\end{figure}

%% file: methods.tex
\subsection{A Bayesian Inference Perspective on Action Selection}
To move beyond the rigidity of hard constraints, we require a principled yet flexible mechanism for action selection. Instead of relying on deterministic rules, our approach need be able to weigh and combine multiple signals: strategic advice, quality predictions and value estimates. Bayesian inference provides the ideal mathematical foundation for this task.

Following the paradigm of Control as Inference \cite{levine2018reinforcement}, we introduce a binary latent variable $O \in \{0, 1\}$ representing the optimality of an action. The goal is to learn a policy $\pi(a_t|s_t)$ that models the posterior probability of an action being optimal, $p(a_t|s_t, O=1)$. Using Bayes' rule, this posterior can be decomposed:
\begin{equation}\label{eq:bayes_rule}
\begin{aligned}
    p(a_t|s_t, O=1) &= \frac{p(O=1|s_t, a_t) p(a_t|s_t)}{p(O=1|s_t)} \\ &\propto \underbrace{p(O=1|s_t, a_t)}_{\text{Likelihood}} \cdot \underbrace{p(a_t|s_t)}_{\text{Prior}}
\end{aligned}
\end{equation}
This formulation provides the theoretical scaffold for SoftPipe. The framework's components are designed to instantiate the \textit{prior} and the \textit{likelihood} terms as described next.

\subsubsection{The LLM as a Strategic Planner}
The prior distribution, $p(a_t|s_t)$, should encapsulate domain knowledge about which actions are sensible \textit{before} considering any task-specific rewards. We select a Large Language Model (LLM) to serve as this knowledge-rich prior due to its remarkable advantage: having been trained on vast corpora of text and code, possess a high-level understanding of common operator semantics. This makes them as a strategic prior that can guide exploration effectively from the very beginning.

The LLM acts as a strategic planner, assessing the current state $s_t$ (represented by meta-features) to produce a probability distribution $m \in \mathcal{M}$:
\[
P_{\text{LLM}}(m|s_t) = \text{LLM-Planner}(s_t)
\]
Where we define the prior probability for any specific action $a_t$ by marginalizing it belongs to:
\begin{equation}
p(a_t|s_t) \triangleq P_{\text{LLM}}(a_t|s_t) = \sum_{m \in \mathcal{M}} P_{\text{LLM}}(m|s_t) \cdot \mathbb{I}[a_t \in \mathcal{A}_m]
\label{eq:prior_def}
\end{equation}
Where $\mathbb{I}[\cdot]$ is the indicator function. This probabilistic guidance from the LLM is the cornerstone of our ``soft guidance''. 

\subsubsection{The Likelihood from RL and LTR}
The likelihood term, $p(O=1|s_t, a_t)$, quantifies the probability that action $a_t$ is optimal. SoftPipe models this likelihood by combining two complementary signals:
\begin{itemize}
    \item \textit{\underline{Long-Term Value (RL)}:} The agent's own action-value function, $Q(s_t, a_t)$, learned through reinforcement learning, estimates the long-term cumulative reward.
    \item \textit{\underline{Immediate Quality (LTR)}:} 
 While the LLM provides strategic guidance and the RL agent learns long-term value, we need a mechanism for immediate, fine-grained quality assessment. We fulfill this role with a pre-trained Learning-to-Rank (LTR) model. It provides a score,  $r_{\text{LTR}}(a_t|s_t)$, predicting an operator's immediate utility based on a large offline dataset of past experiences(represented by \textit{meta-feature, action, performance tuples}). While this provides a crucial ``warm start'', the model is unaware of long-term consequences, necessitating its fusion with the RL agent's Q-function.
\end{itemize}
We combine these to form the log-likelihood: $\log p(O=1|s_t, a_t) \propto \alpha Q(s_t, a_t) + \gamma r_{\text{LTR}}(a_t|s_t)$. This implies the likelihood is an exponential function of their weighted sum:
\begin{equation}
    p(O=1|s_t, a_t) \propto \exp\left(\alpha Q(s_t, a_t) + \gamma r_{\text{LTR}}(a_t|s_t)\right)
    \label{eq:likelihood_def}
\end{equation}
Here, $\alpha$ and $\gamma$ are hyperparameters that control the relative importance of long-term experience versus supervised immediate-quality signals in determining an action's optimality.

\subsection{A Collaborative Policy With Convergence Guarantee}
By substituting our defined prior \eqref{eq:prior_def} and likelihood \eqref{eq:likelihood_def} into the Bayesian posterior formula \eqref{eq:bayes_rule},  we derive the unnormalized log-posterior (logit) for each action $a_t$:
\begin{equation}\label{eq:logit}
\begin{aligned}
    \text{logit}(a_t) &= \log [p(O=1|s_t, a_t) \cdot p(a_t|s_t)] \\ 
    &\propto \log[\exp(\alpha Q(s_t, a_t) + \gamma r_{\text{LTR}}(a_t|s_t))] \\ 
    &\quad+ \log(P_{\text{LLM}}(a_t|s_t)) \\ 
    &= \alpha Q(s_t, a_t) + \gamma r_{\text{LTR}}(a_t|s_t) + \beta \log(P_{\text{LLM}}(a_t|s_t))
\end{aligned}
\end{equation}
where we introduce $\beta$ as a hyperparameter to control the strength of the prior. Applying the softmax function to these logits yields the final, normalized policy for SoftPipe:
\begin{equation}\label{eq:cogniq_policy}
\begin{aligned}
    \pi(a_t|s_t) = \text{softmax} ( \alpha Q(s_t, a_t) &+ \beta \log P_{\text{LLM}}(a_t|s_t) \\&+ \gamma r_{\text{LTR}}(a_t|s_t) )
\end{aligned}
\end{equation}
This theoretically-grounded policy has an interpretation: it is the optimal solution to a KL-regularized objective \cite{yan2024efficient, rafailov2023direct}:
\begin{equation}\label{eq:cogniq_kl}
\begin{aligned}
    \pi^* = \arg\max_\pi & \mathbb{E}_{a \sim \pi(\cdot|s)} \left[ \alpha Q(s,a) + \gamma r_{\text{LTR}}(s,a) \right] \\
    &- \beta \text{KL}\left[\pi(\cdot|s) \,||\, P_{\text{LLM}}(\cdot|s)\right]
\end{aligned}
\end{equation}
This means the agent learns to maximize a composite reward (from RL and LTR) while being regularized to stay close to the LLM's strategic prior. We now establish the theoretical convergence guarantee for this policy:  

\begin{theorem}[Convergence of Collaborative Policy]
\label{thm:convergence}
Let $\pi(a_t|s_t)$ be the SoftPipe policy defined in Equation (6). Under standard RL assumptions (bounded rewards, positive LLM prior $P_{LLM}(a|s) > 0$), the policy converges to the optimal solution of:
\begin{align}
\max_{\pi} \mathbb{E}_{\pi}\bigg[\sum_{t} \big(\alpha Q(s_t, a_t) + \gamma r_{LTR}(s_t, a_t)\big)\bigg] \nonumber \\
- \beta \mathbb{E}_{\pi}\bigg[\sum_{t} KL[\pi(\cdot|s_t) || P_{LLM}(\cdot|s_t)]\bigg]
\end{align}
\end{theorem}

\begin{proof}
The SoftPipe policy can be written as Equation~\ref{eq:cogniq_policy}.
This is the exact solution to the KL-regularized RL objective. To see this, consider the optimization problem for a single state $s$:
\begin{align}
\max_{\pi} \sum_a \pi(a|s)[\alpha Q(s,a) + \gamma r_{LTR}(s,a)] \nonumber \\
- \beta KL[\pi(\cdot|s) || P_{LLM}(\cdot|s)]
\end{align}

Taking the derivative and solving yields:
\begin{equation}
\pi^*(a|s) \propto P_{LLM}(a|s) \exp\left(\frac{\alpha Q(s,a) + \gamma r_{LTR}(s,a)}{\beta}\right)
\end{equation}

which matches our policy formulation. As the Q-function converges through RL updates, the policy converges to this optimal form.
\end{proof}

\subsection{SoftPipe Algorithm}

Algorithm~\ref{alg:SoftPipe} presents the core Collaborative action selection process of SoftPipe, structured as follows: Lines 1-4 initialize the three key components and extract the LLM strategic prior over operator types; Lines 5-9 compute empirical estimators by evaluating both long-term RL values and immediate LTR quality scores for all available operators; Lines 10-14 perform Bayesian fusion by computing logits that combine all three information sources and converting them to a normalized policy distribution; Lines 15-16 sample the final action from this posterior distribution. While the complete end to end train procedure is detailed in Algorithm~\ref{alg:SoftPipe-simple} in Appendix A.1.

\begin{algorithm}[h]
\caption{SoftPipe: Collaborative Action Selection}
\label{alg:SoftPipe}
\begin{algorithmic}[1]
\REQUIRE State $s_t$, Q-function $Q$, LTR model $f_{LTR}$, LLM planner
\REQUIRE Hyperparameters $\alpha$, $\beta$, $\gamma$
\ENSURE Selected action $a_t$

\STATE Step 1: Get LLM Strategic Prior
\STATE Query LLM with meta-features of $s_t$ to get type distribution
\STATE $P_{LLM}(m|s_t) \leftarrow$ LLM-Planner($s_t$) \COMMENT{Probability over operator types}
\STATE $P_{LLM}(a|s_t) \leftarrow \sum_{m \in \mathcal{M}} P_{LLM}(m|s_t) \cdot \mathbb{I}[a \in \mathcal{A}_m]$ \COMMENT{Using Eq. (3)}

\STATE Step 2: Compute Empirical Estimators
\FOR{each operator $a \in \mathcal{O}$}
    \STATE $q_a \leftarrow Q(s_t, a)$ \COMMENT{Long-term value from RL}
    \STATE $r^{LTR}_a \leftarrow f_{LTR}(s_t, a)$ \COMMENT{Immediate quality score}
\ENDFOR

\STATE Step 3: Bayesian Fusion via Softmax
\FOR{each operator $a \in \mathcal{O}$}
    \STATE $\text{logit}(a) \leftarrow \alpha \cdot q_a + \gamma \cdot r^{LTR}_a + \beta \cdot \log P_{LLM}(a|s_t)$
\ENDFOR
\STATE Compute policy $\pi(a|s_t) \leftarrow \text{softmax}(\text{logits})$ \COMMENT{Using Eq. (6)}

\STATE Step 4: Sample Action from Posterior
\STATE Sample action $a_t \sim \pi(\cdot|s_t)$
\RETURN $a_t$
\end{algorithmic}
\end{algorithm}

\subsection{Optimality Gap Analysis}
\label{sec:theory}

A key theoretical justification for SoftPipe is its ability to discover superior solutions that are inaccessible to hard-constraint methods. The following theorem quantifies this advantage by providing a lower bound on the performance improvement over policies restricted by hard constraints.

\begin{theorem}[Optimality Gap Bound]
\label{thm:gap}
Let $\mathcal{O}_c(s) \subset \mathcal{O}$ be the constrained operator set under hard constraints (e.g., CtxPipe's restriction), and $\mathcal{O}$ be the full operator library accessible to SoftPipe. Define:
\begin{itemize}
    \item $V_{hard}^*(s) = \max_{a \in \mathcal{O}_c(s)} Q^*(s,a)$: best value under constraints
    \item $V^*(s) = \max_{a \in \mathcal{O}} Q^*(s,a)$: best value without constraints
\end{itemize}

Then SoftPipe achieves value $V_{SoftPipe}(s)$ satisfying:
\begin{equation}
V_{SoftPipe}(s) - V_{hard}^*(s) \geq \frac{c}{1 + e^{-\Delta/\beta}} \cdot \Delta
\end{equation}
where $\Delta = V^*(s) - V_{hard}^*(s) \geq 0$ is the constraint-induced gap, and $c > 0$ is a constant depending on $|\mathcal{O}|$ and the LLM prior.
\end{theorem}


\begin{proof}  $o^* = \arg\max_{o \in \mathcal{O}} Q^*(s,o)$, the optimal operator is excluded by hard constraints (i.e., $o^* \notin \mathcal{O}_c(s)$), SoftPipe can still select it with probability:
\begin{equation}
\pi(o^*|s) \geq \frac{P_{LLM}(o^*|s)}{Z} \exp\left(\frac{\alpha Q^*(s,o^*)}{\beta}\right)
\end{equation}
where $Z$ is the normalization constant. Since $P_{LLM}(o^*|s) > 0$, this probability is positive.

The expected value under SoftPipe is:
\begin{equation}
V_{SoftPipe}(s) = \sum_{o \in \mathcal{O}} \pi(o|s) Q^*(s,o) \geq \pi(o^*|s) Q^*(s,o^*)
\end{equation}


Next, we need to bound $\pi(o^*|s)$ from below. From Equation~\ref{eq:cogniq_policy} and Equation~\ref{eq:logit}:
\begin{equation}
\pi(o^*|s) = \frac{\exp(logit(o^*))}{\sum_{o' \in \mathcal{O}} \exp(logit(o'))}
\end{equation}

For simplicity, assume $r_{LTR}$ provides unbiased estimates. Then:
\begin{equation}
\pi(o^*|s) \geq \frac{P_{LLM}(o^*|s) \exp(\alpha Q^*(s,o^*)/\beta)}{|\mathcal{O}| \cdot \max_{o'} P_{LLM}(o'|s) \exp(\alpha Q^*(s,o')/\beta)}
\end{equation}

Using the fact that $Q^*(s,o^*) - Q^*(s,o_{hard}^*) = \Delta$ and defining $c = \frac{P_{LLM}(o^*|s)}{|\mathcal{O}| \cdot \max_{o'} P_{LLM}(o'|s)}$:
\begin{equation}
\pi(o^*|s) \geq \frac{c \cdot \exp(\alpha\Delta/\beta)}{1 + \exp(\alpha\Delta/\beta)} = \frac{c}{1 + e^{-\alpha\Delta/\beta}}
\end{equation}
Substituting back:
\begin{equation}
V_{SoftPipe}(s) - V_{hard}^*(s) \geq \frac{c}{1 + e^{-\alpha\Delta/\beta}} \cdot \Delta
\end{equation}

For notational simplicity, we can absorb $\alpha$ into $\beta$ to get the stated result.
\end{proof}

%% file: experiments.tex
\subsection{Experimental Setup} 
\noindent\textbf{Hardware and OS.} We run our experiments on a server with 128 AMD EPYC 7543 CPUs, each with 32 cores, and 256GB memory in total. The local LLMs run on NVIDIA A100 with CUDA 12.8. The OS is Ubuntu 20.04 with Linux kernel 5.15.0-138-generic. The Python version of our experiment is 3.12.9. 

 \noindent\textbf{Datasets.} We evaluate all methods on 18 diverse, real-world datasets from the OpenML repository~\cite{vanschoren2014openml} used by DiffPrep~\cite{li2023diffprep}. These datasets are commonly used benchmarks in AutoML research and span a variety of domains, sizes, and complexities~\cite{li2021cleanml,ledell2020h2o}. The code and dataset are available at github.com/cogniqh0719/project-SoftPipe.

 \noindent\textbf{Baselines.} We compare SoftPipe against a strong set of baselines representing different AutoDP methods:
 \begin{itemize}
    \item \textbf{Standard RL}:  A standard DQN algorithm, and a standard Q-Learning algorithm (abbreviated as QL). Using an $\epsilon$-greedy exploration strategy without any external guidance.
    \item \textbf{RL with Hard Constraints}: CtxPipe~\cite{ctxpipe}, the current state-of-the-art method, and HAIPipe~\cite{chen2023haipipe}). These methods embody the hard-constraint paradigm we critique, using pre-defined rules to restrict the action space.
    \item \textbf{Rule-based Hierarchical RL (HRL)}: A Hard Hierarchical RL agent (Options-DQN) where a high-level policy's choice of an operator type creates a hard commitment, restricting the low-level policy to operators only within that type.
    \item \textbf{LLM-Planner}: Zero-shot pipeline generation using Llama-3, Qwen3 \cite{yang2025qwen3} and GPT-4o. These represent non-RL baselines that lack a closed-loop feedback mechanism~\cite{zheng2024llamafactory}.
    \item \textbf{Classic AutoML Frameworks:} Auto Sklearn~\cite{10.5555/2969442.2969547} uses Bayesian optimization (abbreviated as AutoSk) and TPOT~\cite{olson2016evaluation} uses evolutionary algorithms. 
\end{itemize}

\begin{table*}[h!]
\centering
\caption{Main performance comparison on 18 datasets. SoftPipe achieves the highest average performance and the best (lowest) average rank, demonstrating its consistent superiority over a wide range of baselines. The best result on each dataset is in \textbf{bold}.}
\label{tab:main-results}
\small{
\begin{tabular}{l|cccccccccc|c}
\toprule
\textbf{Dataset} &
  \textbf{QL} &
  \textbf{DQN} &
  \textbf{CtxPipe} &
  \textbf{HAI-AI} &
  \textbf{HRL} &
  \textbf{Llama3.3} &
  \textbf{Qwen3} &
  \textbf{GPT-4o} &
  \textbf{AutoSk} &
  \textbf{TPOT} &
  \textbf{SoftPipe} \\ \midrule
abalone        & 0.274          & 0.257 & \textbf{0.287} & 0.260          & 0.263 & 0.246 & 0.258          & 0.267 & 0.258          & 0.255 & 0.274          \\
ada\_prior     & 0.831          & 0.812 & 0.818          & 0.801          & 0.841 & 0.794 & 0.828          & 0.839 & 0.845          & 0.809 & \textbf{0.846} \\
avila          & 0.682          & 0.620 & 0.759          & 0.630          & 0.486 & 0.613 & 0.542          & 0.609 & 0.701          & 0.604 & \textbf{0.866} \\
connect-4      & 0.747          & 0.664 & 0.763          & 0.775          & 0.664 & 0.657 & 0.665          & 0.713 & \textbf{0.797} & 0.661 & 0.795          \\
eeg            & 0.632          & 0.613 & 0.740          & 0.556          & 0.642 & 0.631 & 0.556          & 0.623 & \textbf{0.917} & 0.585 & 0.837          \\
google         & 0.633          & 0.594 & 0.590          & 0.550          & 0.585 & 0.582 & 0.645          & 0.548 & 0.604          & 0.606 & \textbf{0.666} \\
house          & \textbf{0.938} & 0.908 & 0.818          & 0.928          & 0.918 & 0.918 & 0.914          & 0.904 & 0.927          & 0.928 & 0.932          \\
jungle\_chess  & 0.814          & 0.697 & \textbf{0.861} & 0.760          & 0.674 & 0.677 & 0.677          & 0.760 & 0.815          & 0.664 & \textbf{0.861} \\
micro          & 0.613          & 0.589 & 0.605          & \textbf{0.633} & 0.571 & 0.609 & 0.578          & 0.567 & 0.591          & 0.578 & 0.625          \\
mozilla4       & 0.933          & 0.871 & \textbf{0.940} & 0.870          & 0.860 & 0.915 & 0.921          & 0.867 & \textbf{0.940} & 0.910 & \textbf{0.940} \\
obesity        & 0.790          & 0.825 & 0.868          & 0.768          & 0.726 & 0.759 & 0.681          & 0.851 & 0.830          & 0.808 & \textbf{0.898} \\
page-blocks    & \textbf{0.973} & 0.952 & 0.965          & 0.935          & 0.954 & 0.949 & \textbf{0.973} & 0.955 & \textbf{0.973} & 0.962 & 0.965          \\
pbcseq         & 0.710          & 0.735 & \textbf{0.805} & 0.733          & 0.720 & 0.730 & 0.692          & 0.704 & 0.754          & 0.724 & 0.753          \\
pol            & 0.981          & 0.796 & 0.949          & 0.916          & 0.890 & 0.852 & 0.811          & 0.921 & \textbf{0.984} & 0.981 & 0.981          \\
run\_or\_walk  & 0.945          & 0.972 & 0.956          & 0.915          & 0.834 & 0.837 & 0.676          & 0.718 & \textbf{0.984} & 0.934 & 0.976          \\
shuttle        & 0.995          & 0.965 & \textbf{1.000} & 0.951          & 0.980 & 0.972 & 0.976          & 0.994 & 0.997          & 0.965 & \textbf{1.000} \\
uscensus       & 0.826          & 0.832 & 0.845          & 0.807          & 0.846 & 0.818 & 0.848          & 0.841 & \textbf{0.851} & 0.814 & 0.846          \\
wall-robot & 0.924          & 0.853 & 0.946          & 0.896          & 0.864 & 0.843 & 0.815          & 0.812 & 0.951          & 0.941 & \textbf{0.961} \\ \midrule
AVERAGE        & 0.791          & 0.753 & 0.806          & 0.760          & 0.740 & 0.745 & 0.725          & 0.750 & 0.818          & 0.763 & \textbf{0.835} \\
RANK           & 4.28           & 7.33  & 3.72           & 7.17           & 7.50  & 8.17  & 7.44           & 7.44  & 2.72           & 7.22  & \textbf{1.78} \\ \bottomrule
\end{tabular}}
\end{table*}

\noindent\textbf{Implementation Details.} For RL-based methods, we use a Q-Learning architecture where the action space consists of 24 available operators and a DQN architecture with a 3-layer MLP for the Q-network and a replay buffer of size 5,000. The LLM used is Llama3.3, prompted with a structured description. The LTR model is a LightGBM ranker trained on an offline dataset of (state, action, performance) tuples collected from a wide range of previous experiments. The operator library consists of 24 common operators for imputation, encoding, scaling, feature preprocessing, selection, and engineering, listed in Table~\ref{tab:components}. The downstream ML model for evaluation is Logistic Regression.

\noindent\textbf{Hyperparameters.}
To ensure the reproducibility of our results, we provide the hyperparameter settings for SoftPipe. The key parameters governing the synergy between the LLM prior, LTR scores, and RL values were determined through one-variable-at-a-time strategy on a hold-out set of validation datasets, aiming to maximize the average validation performance. The standard reinforcement learning parameters were set to commonly used values in the literature. All detail settings and analysis are summarized as Table~\ref{tab:hyperparams} in Appendix A.3 and A.4.


\noindent\textbf{Evaluation Metrics.} For the primary task, we report the prediction accuracy of a Logistic Regression model trained on the datasets transformed by the generated pipelines. We also report the average accuracy \textit{ranking} of each method across all datasets. To evaluate efficiency, we measure \textit{inference time}, \textit{pipeline length}. Especially, we plot the best performance found versus the number of pipelines evaluated. 

\subsection{Overall Performance Comparison}
Table \ref{tab:main-results} presents the main experimental results, comparing SoftPipe against all baselines across the 18 datasets.

The results clearly demonstrate the superiority of SoftPipe. With an average score of \textbf{0.834} and an average rank of \textbf{1.78}, it significantly outperforms all other methods. Several key observations can be made:
\begin{itemize}
    \item \textit{\underline{Superiority over Standard RL}}: SoftPipe surpasses standard RL methods like Q-learning (0.791) and DQN (0.753). This confirms that its guided exploration mechanism is far more effective than the undirected $\epsilon$-greedy strategy, which often struggles in the vast search space.
    \item \textit{\underline{Superiority over Hard-Constraint RL}}: SoftPipe outperforms CtxPipe and HAI-AI. This empirically validates our central claim: the flexibility to escape hard constraints allows SoftPipe to discover better pipelines.
    \item \textit{\underline{Advantage over Hierarchy RL}}: SoftPipe (Rank 1.78) also shows a clear advantage over HRL baseline (Rank 7.50). This empirically validates our central claim that the ``soft guidance'' avoids the suboptimal commitments of a rigid hierarchical structure, allowing it to find better solutions.
    \item \textit{\underline{Effectiveness beyond Pure LLM}}: While LLMs provide a strong basis for reasoning, SoftPipe's adaptive, experience-driven approach outperforms zero-shot LLM agents like GPT-4o (Rank 7.44), Llama3.3 (Rank 8.17), Qwen3 (Rank 7.44). This highlights the importance of grounding high-level knowledge with concrete, task-specific feedback from RL and LTR.
\end{itemize}
\begin{figure*}[htbp]
    \centering
\includegraphics[width=\linewidth]{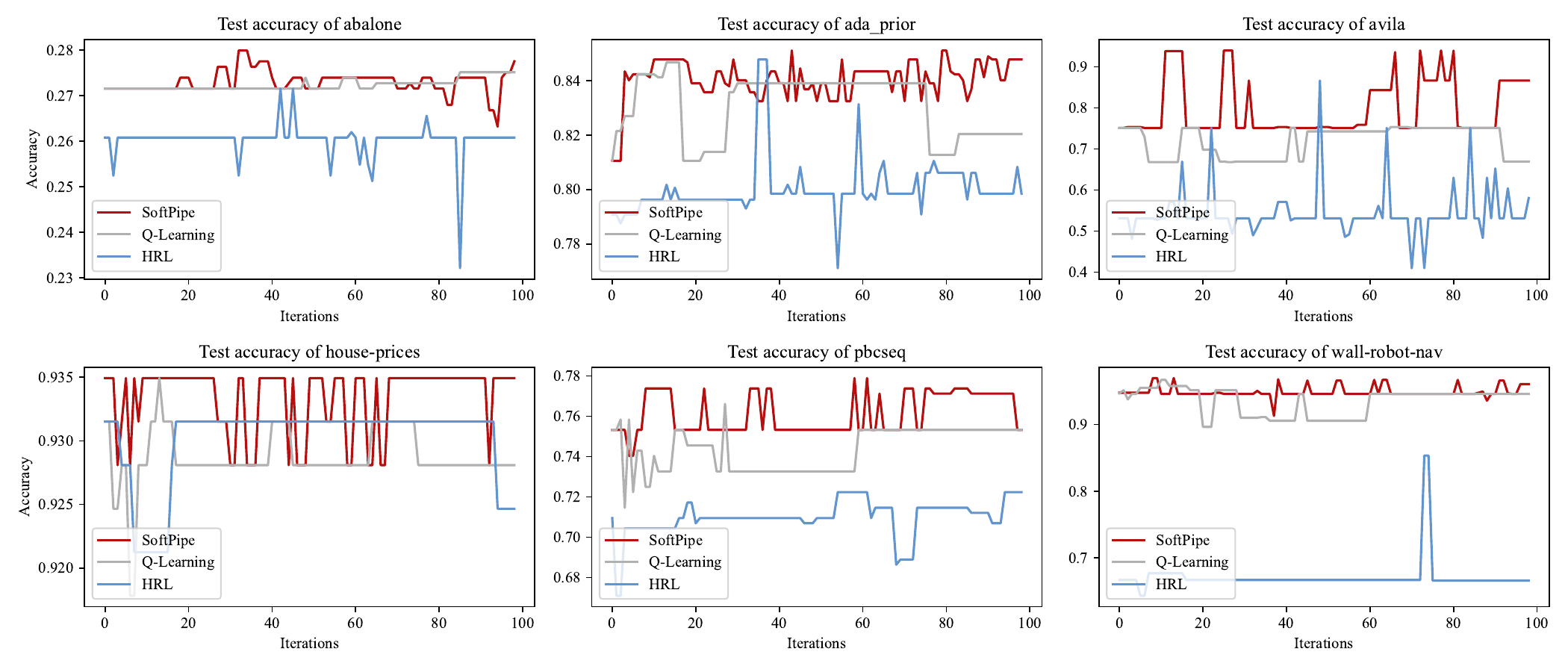}
    \caption{Learning curves on six representative datasets (abalone, ada\_prior, avila, house-prices, pbcseq, wall-robot-nav)}
    \label{fig:learning-curves}
\end{figure*}

\begin{figure}[t]
    \centering
    \includegraphics[width=0.5\linewidth]{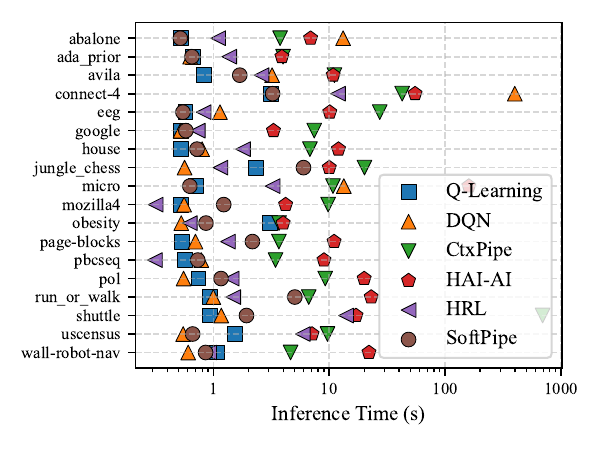}
    \caption{Comparison of inference time} 
    \label{fig:running_time}
\end{figure}

\subsubsection{Learning Efficiency and Pipeline Characteristics}
Figure~\ref{fig:learning-curves} illustrates the learning efficiency of SoftPipe compared to key RL baselines. To provide a clear view of the learning dynamics, we have plotted the performance curves on six representative datasets from our evaluation suite of 18. The learning curve for SoftPipe consistently shows a much steeper initial ascent, which can be attributed to two core components of our framework: the LTR model provides a dense, immediate reward signal, giving the agent a ``warm start,'' while the LLM prior guides exploration towards promising regions from the very beginning. In contrast, both the standard Q-Learning and the HRL baselines exhibit slower learning and often converge to suboptimal solutions. 

\noindent\textbf{Inference Time and Pipeline Length.} Figure~\ref{fig:running_time} shows that SoftPipe's inference time is competitive with other baseline methods. Furthermore, this competitiveness is underscored by the resulting pipelines. As shown in Table~\ref{tab:pipeline-length}, SoftPipe discovers substantially shorter pipelines, with an average length of only \textbf{2.72} operators.  Shorter pipelines are not only more interpretable but also faster to execute and evaluate during the search process itself. 

\begin{table}[t]
    \centering
    \caption{Comparison of average pipeline length of SoftPipe and baselines}
    \begin{tabular}{cc|cc}
\toprule
\textbf{Method} & \textbf{Avg Pipe Len} & \textbf{Method} & \textbf{Avg Pipe Len}  \\
\midrule 
SoftPipe    & 2.72  & CtxPipe   & 6.00 \\
Q-Learning   & 2.56  & DQN       & 3.33 \\
HRL         & 6.00  & Llama3.3  & 4.11 \\
TPOT        & 2.78  & Qwen3     & 3.72 \\
\bottomrule
    \end{tabular}
    \label{tab:pipeline-length}
\end{table}

\subsection{Ablation Study.}
To isolate the contribution of each component within the SoftPipe framework, we conduct a thorough ablation study. Based on the experimental design details in the supplementary materials, we evaluate several variants of our model where one or more components are removed. The results, averaged across all datasets, are presented in Table \ref{tab:ablation}.

\begin{table}[h!]
\centering
\caption{Ablation study results. The full SoftPipe model achieves the best accuracy and rank. Removing any component leads to a performance drop, confirming their individual and collaborative contributions.}
\label{tab:ablation}
\begin{tabular}{lcc}
\toprule
\textbf{Setting} & \textbf{Accuracy ↑} & \textbf{Rank ↓} \\
\midrule
RL (Baseline) & 0.791 & 4.28 \\
\midrule
$P_{\text{LLM}}$ (w/o LTR) & 0.801 & 3.72 \\
LTR (w/o $P_{\text{LLM}}$) & 0.804 & 2.83 \\
\midrule
\textbf{SoftPipe (Full Model)} & \textbf{0.834} & \textbf{1.78}* \\
\bottomrule
\end{tabular}
\end{table}

%% file: Appendix.tex
\section{Appendix A}


\subsection{A.1 Algorithm Pseudocode}

Algorithm \ref{alg:SoftPipe-simple} outlines the core logic of our SoftPipe framework. The algorithm operates over a series of episodes, where each episode is dedicated to constructing and evaluating a single data preparation pipeline (Line 4). 
Within each episode, the agent sequentially builds a pipeline of a predefined length $T$ (Line 8). The cornerstone of our approach lies in the synergistic action selection process at each step. Instead of relying on a single heuristic, SoftPipe computes a unified policy distribution by integrating three distinct signals as formulated in Equation (5): the agent's experience-driven Q-value, the LTR model's immediate quality score, and the LLM's strategic prior (Line 10). An action $a_t$ is then \textit{sampled} from this resulting policy (Line 11), embodying the ``soft'' hierarchical paradigm that allows for flexible, yet guided, exploration.
Once a full pipeline $\pi_{\text{ep}}$ is constructed, it is evaluated on a validation set to obtain a single terminal reward $R$ that reflects its overall quality (Line 20). This holistic reward is then used to update the Q-values for \textit{all} state-action pairs that constituted the pipeline's trajectory (Lines 25-27). This Monte Carlo-style update mechanism enables the agent to learn the long-term consequences of its decisions by associating a sequence of actions with its ultimate outcome. Throughout the training process, the algorithm tracks and retains the best-performing pipeline discovered (Lines 21-23), which is returned as the final output (Line 30).

\begin{algorithm}
\caption{SoftPipe}
\label{alg:SoftPipe-simple}
\begin{algorithmic}[1]
\STATE \textbf{Input:} Raw dataset $D_{\text{raw}}$, max pipeline length $T$, number of episodes $N_{\text{episodes}}$
\STATE \textbf{Initialize:} Q-function $Q$, Pre-trained LTR model, LLM Planner
\STATE \textbf{Initialize:} Best found pipeline $\pi^* \leftarrow \emptyset$, best performance $P^* \leftarrow -\infty$

\FOR{episode $= 1$ \TO $N_{\text{episodes}}$}
    \STATE $s \leftarrow \text{GetState}(D_{\text{raw}})$ \COMMENT{Extract meta-features from data}
    \STATE Initialize current pipeline $\pi_{\text{ep}} \leftarrow \emptyset$
    \STATE Store trajectory for this episode $\text{traj} \leftarrow []$
    
    \FOR{$t = 0$ \TO $T-1$}
        \STATE \COMMENT{--- Synergistic Action Selection ---}
        \STATE Compute policy $\pi_{\text{policy}}(\cdot|s)$ using Eq. (5) (integrating $Q$, LTR, and LLM)
        \STATE Sample action $a_t \sim \pi_{\text{policy}}(\cdot|s)$
        
        \STATE \COMMENT{--- Execute and Record ---}
        \STATE $s_{\text{next}} \leftarrow \text{ApplyAction}(s, a_t)$ \COMMENT{Get next state after applying action}
        \STATE $\pi_{\text{ep}} \leftarrow \pi_{\text{ep}} \oplus a_t$ \COMMENT{Append action to pipeline}
        \STATE $\text{traj.append}((s, a_t))$ \COMMENT{Store state-action pair}
        \STATE $s \leftarrow s_{\text{next}}$
    \ENDFOR
    
    \STATE \COMMENT{--- Evaluate and Learn from the Completed Pipeline ---}
    \STATE Evaluate final pipeline performance: $R \leftarrow \text{EvaluatePipeline}(\pi_{\text{ep}})$
    
    \IF{$R > P^*$}
        \STATE $P^* \leftarrow R$
        \STATE $\pi^* \leftarrow \pi_{\text{ep}}$
    \ENDIF
    
    \STATE \COMMENT{Update Q-values for the episode's trajectory based on the final reward}
    \FOR{each state-action pair $(s_k, a_k)$ in $\text{traj}$}
        \STATE Update $Q(s_k, a_k)$ to better predict the final reward $R$.
        \STATE (e.g., using Monte Carlo update: $Q(s_k, a_k) \leftarrow Q(s_k, a_k) + \eta(R - Q(s_k, a_k)))$
    \ENDFOR
\ENDFOR

\STATE \textbf{Return:} Best found pipeline $\pi^*$
\end{algorithmic}
\end{algorithm}

\subsection{A.2 LLM Prompt Design}
\label{app:prompt_design}
\subsubsection{A Concrete Prompt Example}
Here is a concrete example of a filled-out prompt for the \texttt{wall-robot-nav} dataset at an early stage of pipeline construction. This demonstrates how the abstract template is instantiated with real data and retrieved experiences.

\begin{tcolorbox}[
    colback=promptbackground,      %
    colframe=promptborder,         %
    coltitle=black,                %
    fonttitle=\bfseries,           %
    title=LLM Prompt Example: \texttt{wall-robot-nav} Dataset, %
    arc=2mm,                       %
    boxrule=1pt,                   %
    left=6mm,
    right=6mm,
    top=3mm,
    bottom=3mm,
    listing only,                  %
    listing options={
        language=bash,             %
        basicstyle=\ttfamily\small,%
        keywordstyle=\color{blue},
        commentstyle=\color{gray},
        stringstyle=\color{red!60!black}
    }
]

You are an expert data scientist specializing in automated machine learning. Your task is to analyze the provided dataset context and historical information to propose a probability distribution over available action categories. Each pipeline should be a sequence of data processing operators that aims to discover better data preprocessing workflows more efficiently maximize the prediction accuracy of a downstream classification model.

The available categories are: Imputer (category), Imputer (numeric), Encoder, Feature processor, Feature engineering, Feature selection. 

You must provide a dict-like object to indicate the probability of choosing the action category. Prioritize action categories that seem promising for exploration based on your analysis. Assign very low probabilities to action categories that are clearly irrelevant to the dataset.



\textbf{Current situation and context}:

- Task Type: Logistic regression classification

- Key Dataset Statistics:
  
\quad- Size of dataset: 5456 rows, 24 cols
  
\quad- Missing Values: No
  
\quad- Feature Types: All numerical

\quad- Cols with skewed distribution: V1, V3, V4

\quad- Cols with outliers: V0 (10.31\%), V2 (8.50\%), V7 (2.32\%)

- Current Partial Pipeline: [PowerTransformer, RobustScaler, ...]

\textbf{Available operators}: ImputerMean, StandardScaler, QuantileTransformer, PowerTransformer, PCA, 
MinMaxScaler, VarianceThreshold, ...

---

\textbf{[Example 1]}

-  \textit{Context}: A dataset with no missing values but high skewed distribution columns (B3, B9, ...), oulier columns (B1, B2, B5, ...)

-  \textit{Pipeline}: [QuantileTransformer, RandomTrees-Embedding, MinMaxScaler], accuracy 0.92


\textbf{[Example 2]}

- \textit{Context}: A dataset with many correlated numerical features (A5 and A8: correlation 92\%, A5 and A9: correlation 88\%, ...)

-  \textit{Pipeline}: [StandardScaler, PCA], accuracy 0.88


\end{tcolorbox}

\subsection{A.3 Hyperparameter Settings}

\begin{table*}[h!]
\centering
\caption{Hyperparameter settings for the SoftPipe framework.}
\label{tab:hyperparams}
\begin{tabular}{@{}lll@{}}
\toprule
\textbf{Parameter} & \textbf{Value} & \textbf{Description} \\ \midrule
\multicolumn{3}{l}{\textit{SoftPipe Synergistic Policy Parameters (from Eq. 5)}} \\
$\alpha$ (RL weight) & 1.0 & Weight for the experience-based Q-value. \\
$\beta$ (LLM prior weight) & 2.0 & Controls the strength of the LLM's strategic guidance. \\
$\gamma$ (LTR score weight) & 2.0 & Weight for the immediate quality score from the LTR model. \\
\midrule
\multicolumn{3}{l}{\textit{Reinforcement Learning Parameters}} \\
$\eta$ (Learning rate) & 1e-4 & Learning rate for the Adam optimizer of the Q-network. \\
$\gamma_{\text{RL}}$ (Discount factor) & 0.9 & Discount factor for future rewards in Q-learning. \\
Replay buffer size & 5,000 & Maximum number of transitions stored. \\
Batch size & 200 & Number of transitions sampled for each training step. \\
Q-Network Arch. & 3-layer MLP & (256, 128, 64) neurons with ReLU activation. \\
\midrule
\multicolumn{3}{l}{\textit{General Training Parameters}} \\
$N_{\text{episodes}}$ (Episodes) & 100 & Total number of pipelines generated per dataset. \\
$T$ (Pipeline length) & 8 & Fixed length of each data preparation pipeline. \\ \bottomrule
\end{tabular}
\end{table*}


\subsection{A.4 Hyperparameter Analysis}

We investigate the impact of three key hyperparameters: $\alpha$, $\beta$, and $\gamma$, which weight the action exploration in the training stage. Each parameter is varied independently, where $\alpha$ in the range $[1, 5]$ with a step of 1, $\beta$ and $\gamma$ in the range $[0.1, 4.1]$ with a step of 0.2, while the others are fixed at a certain value. The results are collected from the average test accuracy in datasets ada\_prior, uscensus, and wall-robot-nav.

As shown in Figure~\ref{fig:hyp_analysis}, $\alpha$ shows a stable trend, while $\beta$ and $\gamma$ reveal multiple local maxima. These results indicate that $\beta$ and $\gamma$ require careful tuning, and the presence of multiple peaks may stem from interactions between logit components.

\begin{figure*}
    \centering
    \includegraphics[width=\linewidth]{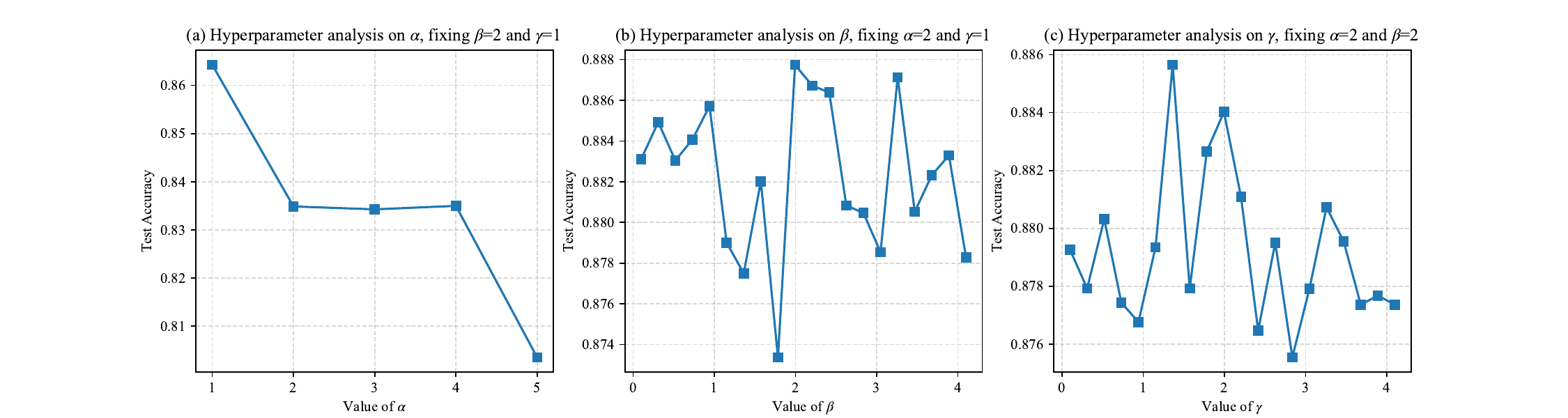}
    \caption{Hyperparameter analysis on the weight of Synergistic Policy}
    \label{fig:hyp_analysis}
\end{figure*}